\documentclass[runningheads,a4paper]{llncs}

\usepackage{amssymb}
\usepackage{graphicx}
\usepackage{epstopdf}
\usepackage{algpseudocode}
\usepackage{algorithm}
\usepackage{placeins}

\usepackage{amsmath,amsfonts,amsthm,amssymb}

\usepackage{xcolor,colortbl}

\definecolor{Gray}{gray}{0.85}
\definecolor{LightCyan}{rgb}{0.88,1,1}
\newtheorem{axiom}{Axiom}
\newcolumntype{a}{>{\columncolor{Gray}}c}
\newcolumntype{b}{>{\columncolor{white}}c}
\usepackage{url}
\urldef{\mailsa}\path|{asad.khan|
\urldef{\mailsb}\path|amir.ali|
\urldef{\mailsc}\path|fauzan.mirza}@seecs.edu.pk|    
\newcommand{\keywords}[1]{\par\addvspace\baselineskip
\noindent\keywordname\enspace\ignorespaces#1 }

\begin{document}

\mainmatter  

\title{CRT and Fixed Patterns in Combinatorial Sequences }

\titlerunning{CRT and Fixed Patterns in Combinatorial Sequences }

%
%
\author{Muhammad Asad Khan%
\and Amir Ali Khan\and Fauzan Mirza}
\authorrunning{M Asad, Amir A Khan, Fauzan Mirza}

\institute{National University of Sciences and Technology,\\
Islamabad, Pakistan\\
\mailsa,
\mailsb,
\mailsc\\}
%
%

\toctitle{Lecture Notes in Computer Science}
\tocauthor{Authors' Instructions}
\maketitle

\begin{abstract}
 In this paper, new context of Chinese Remainder Theorem (CRT) based analysis of combinatorial sequence generators has been presented.  CRT is exploited to establish fixed patterns in LFSR sequences and underlying cyclic structures of finite fields. New methodology of direct computations of DFT spectral points in higher finite fields from known DFT spectra points of smaller constituent fields is also introduced. Novel approach of CRT based structural analysis of LFSR based combinatorial sequence is given both in time and frequency domain. The proposed approach is demonstrated on some examples of combiner generators and is scalable to general configuration of combiner generators.
\end{abstract}

\section{Introduction}
Chinese Remainder Theorem (CRT) is known for centuries as a solution of congruences in number theory  and was appeared in a mathematical classics of Sun Tzu, a mathematician in ancinet China. It is termed as one of the jewels of mathematics and has diverse applications in number theory, abstract algebra, theory of automata, digital signal processing and cryptology. Its magical applications have been classified in three 'C's' which are Computing with various aspects of algorithmics and modular computations,  Theory of Codes and Cryptography~\cite{dingchinese}. From an analytical perspective, CRT is basically a manifestation of addressing complex problems through divide and conquer approach. In other words big structures represented mathematically through their smaller parts mapping the harder problems to their smaller equilvalents and making the analysis easy. In the filed of cryptology, CRT has been known for secret sharing schemes, RSA-CRT and rebalanced RSA-CRT. Continual to new contexts of CRT, new results on applications of CRT in analysis of LFSR based sequneces have been presented in this paper.

 This paper shows that there exist hidden structures in underlying finite fields related to LFSR based combinatorial sequences  which can be exploited through CRT. Number of constituent LFSRs in a combiner generator posses certain fixed patterns in their base finite fields which can be directly mapped through CRT to resultant fields even being combined through non linear functions. These results are consistent both in time and frequency domain. Direct computation of spectral components in higher fields from smaller field spectral components through CRT is yet  a new idea introduced in this paper.  CRT based direct relevance of components of smaller fields to higher fields is novel in associated finite fields theory of combinatorial sequence generators and has obvious usefullness in coding theory and cryptology.

The paper is organized as follows: Section 2 describes the mathematical priliminaries on the subject. In section 3, CRT based fixed patterns existing in the product sequences both in time and frequency domain have been deliberated upon. Section 4 covers the generalized case of combinatorial sequence generators and new methodology to compute spectral components in higher fields from spectral components of contituent fields is given. Comparison of computational complexity of proposed methodology of DFT computataions viz-a-viz classical DFT methods is also included in this section. In Section 5, applications of our results on CRT based fixed structures in cryptanalysis are discussed with small example of a combiner generator. The paper is final concluded in Section 6. 
\section{Mathematical Priliminaries}
Classical theory on LFSR sequences and their applications in cryptology can be found in ~\cite{golomb1982shift}, ~\cite{golomb2005signal} and ~\cite{rueppel1986analysis}. In this section, basic fundamentals related to algebraic theory of LFSR sequences and their frequency domain representaions have been presented. By analyzing the sequences in both time and frequency domain simultaneously, fixed structures related to LFSR sequences and underlying finite fields are highlighted which are considered useful in coding theory and cryptanalysis.

Discrete Fourier Transform (DFT) is considered one of the most important discovery in the area of  signal processing. DFT presents us with an alternate mathematical tool that allows us to examine the frequency domain behaviour of signals, often revealing important information not apparent in time domain. DFT $S_k$ of an n-point sequence $s_i$ is expressed in terms inner product between the sequence and set of complex discrete frequency exponentials:

\begin{equation}\label{DFT-dsp eq}
  S_{k} = \sum_{i=0}^{n-1} s_{i}e^{-j2\pi ik/n} , \;\;\;  k = 0,1,2,.....,n-1
\end{equation}
 The term $e^{-j2\pi ik/n}$ represents discrete set of exponentials. Alternatively, $e^{-j2\pi/n}$ can be viewed as $n^{th}$ root of unity.  


Analogous to the classical DFT, a DFT for a periodic signal $s_t$ with period $n$ defined over a finite field $GF(2^m) $  is represented as

\begin{equation}\label{DFT eq}
  S_{k} = \sum_{t=0}^{n-1} s_{t}\alpha^{tk} , \;\;\;  k = 0,1,2,.....,n-1
\end{equation}
 where $S_k$ is $k$-th frequency component of DFT and $\alpha$ is the primitive element; generator of $GF(2^m) $ with period $n$ ~\cite{pollard1971fast}. For Inverse DFT, we will have a relation 
\begin{equation}\label{IDFT eq}
  s_t = \sum_{k=0}^{n-1} S_{k}\alpha^{-tk} , \;\;\;  k = 0,1,2,.....,n-1
\end{equation}

Similarly for polynomials, we have a relation for DFT and IDFT. Having a correspondence between a minimum polynomial and its associated sequence $s_t$ with $s(x) = \sum_{t=0}^{n-1}s_{t} x^{t}$ and $S(x) = \sum_{k=0}^{n-1}S_{k} x^{k}$, following relation holds for DFT~\cite{golomb2005signal}: 
\begin{equation}\label{DFT eq-poly}
  S_{k} = s(\alpha^{-k}), \;\;\;  k = 0,1,2,.....,n-1
\end{equation}
and similarly for IDFT:
\begin{equation}\label{DFT eq-poly}
  s_t = S(\alpha^{t}), \;\;\;  t = 0,1,2,.....,n-1
\end{equation}
The same sequence $s_t$ can also be expressed in terms of its trace representation~\cite{gong1999transform}; a linear operator from $GF(2^m)$ to its subfiled $GF(2)$. Let $Tr_{1}^{m}(x) = \sum_{k=0}^{m-1} {x^{2}}^{k}$ be the trace mapping from $GF(2^m)$ to $GF(2)$, then $m$ sequence $s_t$ can be represented as:
\begin{equation}\label{eq:trace}
s_t = Tr_{1}^{m} (\beta \alpha^t)
\end{equation}

 where $\alpha$ is a generator of a cyclic group $GF(2^{m})^*$ and is called as primitive element of $GF(2^{m})$. Note that $\beta \in GF(2^{m})$ and each of its nonzero value corresponds to cyclic shift of the $m$-sequence generated by an LFSR with primitive polynomial $f(x)$. Importance of this interpretation of $m$-sequence is that different sequences constructed from root $\alpha$ of primitive polynomial $f(x)$ are cyclic shifts of the same $m$-sequence. The associated linear space $G(f)$ of dimension $m$ contains $2^{m}$ different binary sequences including all 0s sequence as:
 \begin{equation}\label{eq:gf}
G(f) = \{\; \tau^{i}s \; \vert \; 0\; \leq \;i \;\leq \;2^m - 2\; \} \bigcup \{0\}   
\end{equation}
 where $\tau$ is a left shift operator and represents a linear transformation of sequence $s_t$.
According to Blahut's famous theorem, the linear complexity of a peridic sequence over $GF(2^m)$ of period $n$ is equal to the hamming weight of its fourier transform, provided a fourier transform of block length $n$ exists~\cite{Blahut1983errorcontrol}. All DFT components of an LFSR sequence $\in GF(2^m)$.

The zero components in the Fourier spectrum of a sequence over $GF(2^m)$ are related to the roots of a polynomial of that sequence. For example, DFT of an LFSR sequence with feedback polynomial $f(x) = x^3+x+1$ initialized with state $001$ is $0, 0, 0, \alpha^4 , 0, \alpha^2 , \alpha$. As roots of $f(x)$ are $\alpha$ alongwith its conjugates i.e. $\alpha^2$ and $\alpha^4$, so first, second and fourth spectral components are zero. Indices of non zero DFT points for LFSR with minimum polynomial and no multiple roots also follow a fixed pattern. If $k$-th component of spectral sequence is non zero then all $(2^{j}k)\; mod\; n$ components will be harmonics of the $k$-th component where  $1\leq j \leq m-1$.
As DFT of a time domain signal comprises of a fundamental frequency and its harmonics, DFT of an LFSR sequence based on a minimal polynomial with no multiple roots also comprises of $\alpha^i \in GF(2^m)$  and its harmonics $\alpha^{i{^j} mod\mbox{\;} n} \in GF(2^m)$ with $0\leq i \leq n-1$. This harmonic pattern can be efficiently exploited in cryptanalysis attacks on LFSR based sequence generators.

 Let two sequences related by a time shift $u_{t} = s_{t+\tau}$, their DFTs $U_{k}$ and $S_{k}$ are related as:
\begin{equation}\label{DFT shift eq}
  U_{k} = \alpha^{k\tau} S_{k}, \;\;\;\;\;   k = 0,1,....,n-1
\end{equation}

Indices of non-zero spectral points of an LFSR sequence does not change with the shift in LFSR sequence. A non-zero $k$-th component of DFT of an LFSR sequence will always be non-zero. Any shift in LFSR sequence will only change the value at this component by Equation~(\ref{DFT shift eq}). Converse is also true for zero spectral points of an LFSR sequence which will always be zero no matter how much sequence is shifted.

A binary sequence $s_t$ can be represented in terms of trace function with spectral componenets as follows:-

\begin{equation}\label{eq:trace in DFT}
s_t = \sum_{j\in \Gamma(n)} Tr_{1}^{m_j} (A_j \alpha^{-jt}),\;\;\;\;\;t = 0,1,....,n-1
\end{equation}
 where $Tr_{1}^{m_j}$ is a trace function from $GF(2^m)$ to $GF(2)$, $A_j\in GF(2^m)$ and $\Gamma(n)$ is a set of  cyclotomic coset leaders modulo $n$.

\section{CRT and Underlying Finite Field Theory of  Product Sequences} 

In this section,  analysis of a product sequence generated through multiplication of two LFSRs sequences is presented which includes new results on underlying algebraic theory of finite fields. A CRT based linear structure existing in the time and frequency domain representation of the product sequence is presented which renders itself useful for coding theory and cryptanalysis of LFSR based sequence generators.
We build our analysis by starting with a simple case of  multiplication of output sequences of two LFSRs and illustrate  our novel observations on fixed structures existing in the time as well as frequency domian representation of product sequences. The observations of this special case will be generalized to a combinatorial generators in the next section.

\begin{theorem}
\label{CRT-theorem}
Let $s_t\in GF(2^m)$ be a reference product sequence with period $n \mid 2^m - 1$ having two constituent LFSRs defined over primitive polynomials with individual periods $n_1$ and $n_2$. With different shifts $k_1$ and $k_2$ in initials states of LFSRs, resulting output sequnece $u_t$ is correlated to $s_t$ by $u_{t} = s_{t+\tau}$ where shift $\tau$ is determined through CRT as

      \begin{eqnarray*}
      \tau  &\equiv& k_1\mbox{\; (mod} \mbox{\;}n_1) \\
      \tau  &\equiv& k_2\mbox{\; (mod} \mbox{\;}n_2)   
               \end{eqnarray*}

\end{theorem}
\begin{proof}
Within a cyclic group $GF(2^m)$, associated linear space $G(f)$ of dimension $m$ contains $2^{m}-1$ non-zero binary sequences by~(\ref{eq:gf}).
\\ As $s_t$ and $u_t$ both $\in$ $GF(2^m)$, they are shift equilvalents by~(\ref{eq:shift}) with unknown shift value of $\tau$. 
\\ The product sequnec $s_t$ of $a_t$ and $b_t$ can be expressed as
\begin{equation}
s_i = a_{j}.b_{v} 
\end{equation}

 where  $0\leq i \leq n-1$ , $0\leq j \leq n_{1}-1$ and   $0\leq v \leq n_{2}-1$. 
\begin{axiom}
\label{L-1}
While contributing towards a product sequence of length $n$ with two LFSRs, stream of LFSR-1 defined over $GF(2^p)$ with primitive polynomial and its maximum period $2^{p}-1$  is repeated $\delta_1$ times while LFSR-2 defined over $GF(2^q)$ with primitive polynomial as well and corresponding period $2^{q}-1$ is repeated $\delta_2$ where

\begin{eqnarray*}
      \delta_1  &= & \frac{\mbox{lcm}(n_1,n_2)}{n_1}\mbox{\;,\;\; and\;}  \\
      \delta_2  &= & \frac{\mbox{lcm}(n_1,n_2)}{n_2}   
               \end{eqnarray*} 
               \end{axiom}

\begin{axiom}
\label{L-2}
 Within a sequence of period $n$ for a product sequence, each value of index $j$ corresponds to all values of index $v$ if and only if $gcd(n_1,n_2)=1$.
\end{axiom}
From Axioms~\ref{L-1} and~\ref{L-2}, any shift in LFSRs initial states will produce output corresponding to some fixed indices of $j$ and $v$ which already existed in the refernce sequence at some fixed place with initial states of LFSRs without shift.\\
\\ With known values of $j$ and $v$ i.e. $k_{i's}$, CRT will give us the value of $\tau$ mod $n$ as
\begin{eqnarray*}
      \tau  &\equiv& k_1\mbox{\; (mod} \mbox{\;}n_1) \\
     \tau  &\equiv& k_2\mbox{\; (mod} \mbox{\;}n_2)   
               \end{eqnarray*}
\end{proof}
 Let we explain the facts with an example.
\begin{example}\label{E-1}
Let we have a sequence $s_t$ generated from product of two LFSRs having primitive p[olynomials of $g_1(x) =  x^2+x+1$ and $g_2(x)  = x^3+x+1$. The period $n_1$ of stream $a_t$ corresponding to LFSR-1 is $3$ and $n_2$ of $b_t$ corresponding to LFSR-2 is $7$. The period $n$ of $s_t$ is $21$.

Table~\ref{tab:cyclic-m} demonstrates product of two $m$ sequences generated from these two LFSRs. 
 \begin{table}[H]
\small
\begin{center}
\caption[Sample Table]{Product sequence of 2x LFSRs with $n_1 = 3$ and $n_2 = 7$}
\begin{tabular}{|c| c| c| c| c| c| c| c| c| c| c| c| c| c| c| c| c| c| c| c| c| }
\hline
 0 & 1 & 2 & 3 & 4 & 5 & 6 & 7 & 8 & 9 & 10 & 11 & 12 & 13 & 14 & 15 & 16 & 17 & 18 & 19 & 20  \\ \hline
  
$a_{1}$ & $a_{2}$ &$a_{3}$ & $a_{1}$ & $a_{2}$ & $a_{3}$ & $a_{1}$ & $a_{2}$ &$a_{3}$ & $a_{1}$ & $a_{2}$ & $a_{3}$ & $a_{1}$ & $a_{2}$ & $a_{3}$ & $a_{1}$ & $a_{2}$ & $a_{3}$ & $a_{1}$ & $a_{2}$ & $a_{3}$  \\ \hline

 $b_{1}$ & $b_{2}$ &$b_{3}$ & $b_{4}$ & $b_{5}$ & $b_{6}$ & $b_{7}$ &  $b_{1}$ & $b_{2}$ &$b_{3}$ & $b_{4}$ & $b_{5}$ & $b_{6}$ & $b_{7}$ & $b_{1}$ & $b_{2}$ &$b_{3}$ & $b_{4}$ & $b_{5}$ & $b_{6}$ & $b_{7}$  \\ \hline

 $s_{1}$ & $s_{2}$ &$s_{3}$ & $s_{4}$ & $s_{5}$ & $s_{6}$ & $s_{7}$ &  $s_{8}$ & $s_{9}$ &$s_{10}$ & $s_{11}$ & $s_{12}$ & $s_{13}$ & $s_{14}$ & $s_{15}$ & $s_{16}$ &$s_{17}$ & $s_{18}$ & $s_{19}$ & $s_{20}$ & $s_{21}$  \\ \hline
\end{tabular}

\label{tab:cyclic-m}
\end{center}

\end{table}

We analyze the impact of shift on LFSR sequences and their behaviour in cyclic stuctures of finite fields involved. We will shift the LFSR sequences one by one and observe the fixed patterns which can be exploited in cryptanalysis of the combiner generators in particular. We can represent shifts in LFSRs sequences with $k$ and $l$ as 

 \begin{equation}\label{eq:twoLFSR-shift}
s_t =   a_{i+k}.b_{i+l}\;,\;\;\;\;\;\;\;\mbox{with} \; 0\leq \; i \; \leq n-1
\end{equation}

where $k \in [0,n_{1}-1]$ and $l \in [0,n_{2}-1]$. Table~\ref{tab:1-bit-a} demonstrates the scenerio where $a_t$ is left shifted by one bit while keeping the $b_t$ fixed with initial state of '1'.

\begin{table}[H]
\small
\begin{center}
\caption[Sample Table]{Product sequence with $a_{t}$ shifted left}
\begin{tabular}{|c| c| c| c| c| c| c|c| c| c| c| c| c| c| c| c| c| c| c| c| c| }
\hline

 7 & 8 & 9 & 10 & 11 & 12 & 13 & 14 & 15 & 16 & 17 & 18 & 19 & 20 & 0 & 1 & 2 & 3 & 4 & 5 & 6  \\ \hline
  
  $a_{2}$ &$a_{3}$ & $a_{1}$ & $a_{2}$ & $a_{3}$ & $a_{1}$ & $a_{2}$ &$a_{3}$ & $a_{1}$ & $a_{2}$ & $a_{3}$ & $a_{1}$ & $a_{2}$ & $a_{3}$ & $a_{1}$ & $a_{2}$ & $a_{3}$ & $a_{1}$ & $a_{2}$ & $a_{3}$ &  $a_{1}$  \\ \hline

 $b_{1}$ & $b_{2}$ &$b_{3}$ & $b_{4}$ & $b_{5}$ & $b_{6}$ & $b_{7}$ &  $b_{1}$ & $b_{2}$ &$b_{3}$ & $b_{4}$ & $b_{5}$ & $b_{6}$ & $b_{7}$ & $b_{1}$ & $b_{2}$ &$b_{3}$ & $b_{4}$ & $b_{5}$ & $b_{6}$ & $b_{7}$  \\ \hline
 
 $s_{8}$ & $s_{9}$ &$s_{10}$ & $s_{11}$ & $s_{12}$ & $s_{13}$ & $s_{14}$ &  $s_{15}$ & $s_{16}$ &$s_{17}$ & $s_{18}$ & $s_{19}$ & $s_{20}$ & $s_{21}$ & $s_{1}$ & $s_{2}$ &$s_{3}$ & $s_{4}$ & $s_{5}$ & $s_{6}$ & $s_{7}$  \\ \hline
\end{tabular}

\label{tab:1-bit-a}
\end{center}
\end{table}

Comparison of Table~\ref{tab:cyclic-m} with Table~\ref{tab:1-bit-a} reveals that shifting one bit left of $a_t$ and fixing the $b_t$ to reference initial state of '1' shifts $s_t$ by seven units left. Similarly, shifting another bit of $a_t$ to left, brings $a_{3}$ corresponding to $b_{1}$ which can be located in Table~\ref{tab:cyclic-m} at shift position 14. So two left shifts of $a_t$ shifts $s_t$ by 14 units left with reference to bit positions in Table~\ref{tab:cyclic-m}. Now we analyze the impact of left shift of $b_t$ on $s_t$. Table~\ref{tab:1-bit-b} demonstrates the scenerio where $b_t$ is left shifted by one bit while keeping the $a_t$ fixed with initial state of '1'. 
\begin{table}[H]
\small
\caption[Sample Table]{Product sequence with $b_{t}$ shifted left}
\begin{center}
\begin{tabular}{|c| c| c| c| c| c| c| c| c| c| c| c| c| c| c| c| c| c| c| c| c| }
\hline
 15 & 16 & 17 & 18 & 19 &20 &0 & 1 &2 & 3 & 4 & 5 & 6 & 7 & 8 & 9 & 10 & 11 & 12 & 13 & 14 \\ \hline
   $a_{1}$ &  $a_{2}$& $a_{3}$ & $a_{1}$ & $a_{2}$ & $a_{3}$ & $a_{1}$ & $a_{2}$ &$a_{3}$ & $a_{1}$ & $a_{2}$ & $a_{3}$ & $a_{1}$ & $a_{2}$ & $a_{3}$ & $a_{1}$ & $a_{2}$ & $a_{3}$ & $a_{1}$ & $a_{2}$ & $a_{3}$   \\ \hline

  $b_{2}$ &$b_{3}$ & $b_{4}$ & $b_{5}$ & $b_{6}$ & $b_{7}$ &  $b_{1}$ & $b_{2}$ &$b_{3}$ & $b_{4}$ & $b_{5}$ & $b_{6}$ & $b_{7}$ & $b_{1}$ & $b_{2}$ &$b_{3}$ & $b_{4}$ & $b_{5}$ & $b_{6}$ & $b_{7}$ & $b_{1}$  \\ \hline
 
 $s_{16}$ & $s_{17}$ &$s_{18}$ & $s_{19}$ & $s_{20}$ & $s_{21}$ & $s_{1}$ &  $s_{2}$ & $s_{3}$ &$s_{4}$ & $s_{5}$ & $s_{6}$ & $s_{7}$ & $s_{8}$ & $s_{9}$ & $s_{10}$ &$s_{11}$ & $s_{12}$ & $s_{13}$ & $s_{14}$ & $s_{15}$  \\ \hline

\end{tabular}

\label{tab:1-bit-b}
\end{center}
\end{table}

It can be easily seen that one left shift in $b_t$ shifts  $s_t$ by 15 units where $b_{2}$ is corresponding to $a_{1}$. Similarly, another left shift in $b_t$ shifts $s_t$ by another 15 units bringing the $b_{3}$ corresponding to $a_{1}$. Subsequently, three left shifts in $b_t$ with reference to initial state of '1' brings $b_{4}$ corresponding to $a_{1}$ which is at shift index-3 in Table~\ref{tab:cyclic-m}. Similar fixed patterns can be observed for simultaneous shifts of LFSRs and it will be discussed with more detail in following paragraphs.

Let us model this fixed patterns in LFSRs cyclic structures and shifts in intial states of LFSRs through CRT as

	\begin{eqnarray*}
      x  &\equiv& k\mbox{\; (mod} \mbox{\;}n_1) \\
      x  &\equiv& l\mbox{\; (mod} \mbox{\;}n_2)   
               \end{eqnarray*}
	where $k$ and $l$ denote the amount of shifts in initial state of individual LFSRs with reference to initial state of '1'. The solution of CRT i.e. $x$(mod $r$) gives the amount of shift in $s_t$ with reference to $u_t$ as depicted in (\ref{eq:shift}). Consider a scenerio again where $a_t$ is shifted left by one bit and $b_t$ is fixed with initial state of '1' and can be expressed as 
	\begin{eqnarray*}
      x  &\equiv& 1\mbox{\; (mod} \mbox{\;}3) \\
      x  &\equiv& 0\mbox{\; (mod} \mbox{\;}7)   
               \end{eqnarray*}	
	The CRT gives the solution of 7(mod $21$) which is index position of $a_2$ corresponding to $b_1$ in Table~\ref{tab:cyclic-m} shifting the product sequence $s_t$ by seven units left. Consider another scenerio of simultaneous shifts in both LFSRs sequences where $a_t$ is shifted left by one bit and $b_t$ is shifted left by 3 bits with reference to their initial states of '1' and can be expressed as
	\begin{eqnarray*}
      x  &\equiv& 1 \mbox{\; (mod} \mbox{\;}3) \\
      x  &\equiv& 3 \mbox{\; (mod} \mbox{\;}7)   
               \end{eqnarray*}
	The CRT gives value of $-11$ which is $10$ (mod $21$), representing the product sequence $u_t$ as 10 units left shifted version of $s_t$. This value matches to index position of $b_{4}$ corersponding to $a_{2}$ in Table~\ref{tab:cyclic-m}.

Our Observations related to direct correspondence of shift index with initial states of LFSRs and CRT calculations done modulo periods of individual LFSRs are valid for any number of LFSRs in different configurations of nonlinear sequence generators. These observations on classical theory of LFSR cyclic structures with their CRT based interpretation are considered significant for cryptanalysis.  

 \end{example}

In addition to the results of Blahut's theorem on time and frequency domain relationship of sequences, an important corollary establishes new facts related fourier transform in binary fields.
\begin{corollary}
Let $s_t\in GF(2^m)$ be a product sequence with period $n \mid 2^m - 1$ having two constituent sequences $a_t \in GF(2^p)$ and $b_t \in GF(2^q)$ of LFSRs each defined over primitive polynomials with individual periods $n_1 = 2^p - 1$ and $n_2 = 2^q - 1$. If \textbf{$A$} be a DFT spectra of $a_t$, \textbf{$B$} be a DFT spectra of $b_t$ and \textbf{$S$} be a DFT spectra of $s_t$, non zero spectral components of  \textbf{$S$} will only exist at those indices where spectral components of \textbf{$A$} and \textbf{$B$} are non zero.
\end{corollary}
we have another associated corollary here:-
\begin{corollary}\label{cor-index}
With known non zero spectral components of \textbf{$A$} and \textbf{$B$},  non zero spectral components of \textbf{$S$} can be directly determined through Chinese Remainder Theorem (CRT) as:
   
   \begin{eqnarray*}
      x  &\equiv& k_1\mbox{\; (mod} \mbox{\;}n_1) \\
      x  &\equiv& k_2\mbox{\; (mod} \mbox{\;}n_2)   
               \end{eqnarray*}
   where $k_1$ and $k_2$ are non zero index positions of $A_k$ and $B_k$ respectively and $x$ is the position of non zero componenet of DFT spectra of $s_t$ within its period $n$.

\end{corollary}
It is important to observe here that indices of non zero spectral components present in a complete spectrum of resultant stream are determined while working in base fields of component LFSRs and without computing DFT of $s_t$ in a larger field. 
Let we explain these corollaries through a small example here.
\begin{example}\label{exm-2}
Following the Example~\ref{E-1}, consider a product sequence $s_t$ generated from two LFSRs with minimum polynomials  $g_{1}(x) = x^3 + x +1$ and  $g_{2}(x) = x^2 +x +1$.

\begin{enumerate}
\item In time domain representation, we have following sequences.\\
\;\;\;\;\;\;\;\;Sequence $a_t$:\;\;\;$011$\;\;\;\;\;\;\;\;\;\;\;\;\;\;\;\;\;\;\;\;\;\;\;\;\;\;\;\;\;\;\;\;\;\;\;\;\;\;\;\;\;\;\;\;\;\;\;\;\;\;\;\;\;\;\;\;\;\;\;\;\;\;\;\;\;\;\;\;\;\;\;\;\;(of period 3)\\
\;\;\;\;\;\;\;\;	Sequence $b_t$:\;\;\;$0010111$\;\;\;\;\;\;\;\;\;\;\;\;\;\;\;\;\;\;\;\;\;\;\;\;\;\;\;\;\;\;\;\;\;\;\;\;\;\; \;\;\;\;\;\;\;\;\;\;\;\;\;\;\;\;\;\;\;\;\;\;\;\;\;\;(of period 7)\\
\;\;\;\;\;\;\;\;	Sequence $s_t$:\;\;\;$001011000001010010011101110111011101$\;\;\;\;\;\;\;\;\;\;\;\;\;\;(of period 21)\\
\item From~(\ref{DFT eq}), frequency domain representations of these sequences are:
\begin{enumerate}
\item \textbf{$A$} $ = 0,1,1$
\item \textbf{$B$}  $= 0,0,0,\alpha^4, 0,\alpha^2,\alpha$
\item To compute \textbf{$S$}, associated minimum polynomial is determined through Berlekamp-Massey algorithm which is 
$g(x) = x^6+x^4+x^2+x+1$.\\
\textbf{$S$} $= 0,0,0,0,0,\alpha^9,0,0,0,0, \alpha^{18},0,0,\alpha^{15},0,0,0,\alpha^{18},0,\alpha^9,\alpha^{15}$
\end{enumerate}
\end{enumerate} 
\end{example}

Non-zero DFT points in \textbf{$S$} clearly follow a linear behaviour as of time domain representation where any $k$-th component is non-zero if and only if $A_{k}$ and $B_{k}$ are both non-zero. Through non zero indices of \textbf{$A$} and \textbf{$B$}, CRT can be directly used to determine  non-zero spectral points of \textbf{$S$}. For instance,
  \begin{eqnarray*}
      x  &\equiv& 1\mbox{\; (mod} \mbox{\;}3) \\
      x  &\equiv& 3\mbox{\; (mod} \mbox{\;}7)   
               \end{eqnarray*}
               results into index 10 where $\alpha^{18}$ is a non zero spectral component of \textbf{$S$}.
These results on determining non zero spectral indices for product of two sequences are valid for product sequences containing more number of LFSRs as well. 
               
Harmonic pattern of DFT spectra are visible for \textbf{$A$}, \textbf{$B$} and \textbf{$S$}. 
Non-zero indices of DFT sequences also follow a fixed pattern. In case of \textbf{$S$}, non zero DFT element at index 5 has its harmonics at indices $10,\; 20,\; 19 \;(40\; \mbox{mod\;} 21), \;17\;(80\; \mbox{mod\;} 21)$ and at $13\;(160\; \mbox{mod\;} 21)$. The zero components in the fourier transform of a product sequence $s_t$ defined over $GF(2^m)$ are related to roots of $g_(x) = x^6+x^4+x^2+x+1$. As roots of $g(x)$ are $\alpha$ alongwith its conjugates i.e. $\alpha^2$, $\alpha^4$, $\alpha^8$ and $\alpha^{16}$ so first, second, fourth, eigth and sixteenth spectral components are zero.  

\section{Computing the Spectral Components in $GF(2^m)$ through CRT }
Computing DFT of a sequence \textbf{s}$\in GF(2^m)$ by equation~(\ref{DFT eq}) over binary fields requires determining the associated minimum polynomial $m(x)$ of \textbf{$s$}. The most efficient method which computes the linear complexity $l$ of a periodic sequence \textbf{$s$} and gives its minimum polynomial is berlekamp massey algorithm~\cite{golomb2005signal}. The algorithm further requires $2l$ bits of the sequence to determine the linear complexity and minimum polynomial $m(x)$. Based on the root of minimum polynomial $m(x)$, equation~\ref{DFT eq} requires complete period of the sequence to compute each spectral componenet of \textbf{$S$}. Faster method to compute DFT in binary fields proposed in~\cite{gong2011fast} requires lesser number of bits equal to linear complexity $l$ or in few cases lesser than that. However, in all these cases computations have to be in $GF(2^m)$ to which sequence \textbf{$s$} belongs. In this section, new method has been introduced which allows mapping of spectral components of smaller constituent fields to larger finite fields with few limitations of choice of particular indices. We will develop our idea progressively from product of sequences in time domain to a genarlized case of boolean functions where addition of bits in $GF(2)$ is  encompassed as well.
\subsection{Product of Arbitrary Number of $m$-Sequences} In this subsection, case of product sequence is considered where any arbitrary number of LFSR sequences are multiplied togather. Starting with simple case of two LFSRs, we will establish facts for more number of LFSRs where direct computation of spectral points for product sequence is done from DFT points of individual LFSR sequences. We have an important theorem here.
\begin{theorem}
\label{CRT-Mtheorem}
Let $s_t\in GF(2^m)$ be a product sequence with period $n  \mid 2^m - 1$ having $r$ constituent sequences $a_i \in GF(2^{p_i})$ of LFSRs each defined over primitive polynomials with individual periods $n_i = (2^{p_{i}} - 1)$, where all $n_i$ are coprime to each other and $0 \leq i \leq r-1$. Let \textbf{$A^i$} be a DFT spectra of $a_i$, a $k$-th spectral component of \textbf{$S$} corresponding to each non-zero spectral components of \textbf{$A^{i}_{(k\mbox{\;} mod \mbox{\;} n_i)}$} can be determined directly through CRT as

      \begin{eqnarray*}
      d  &\equiv & d_1\mbox{\; (mod} \mbox{\;}n_1) \\
      d  &\equiv & d_2\mbox{\; (mod} \mbox{\;}n_2) \\
      \mbox{..} & \mbox{..} & \mbox{................}\\
      \mbox{..} & \mbox{..} & \mbox{................}\\
       d  &\equiv & d_r\mbox{\; (mod} \mbox{\;}n_r) 
               \end{eqnarray*}
      
   where $d$, $d_1$, $d_2$,..., $d_r$  are degrees of non-zero spectral components i.e. \textbf{$S_k$}, \textbf{$A^{1}_{(k \mbox{\;} mod \mbox{\;} n_1)}$},..., \textbf{$A^{r}_{(k \mbox{\;} mod \mbox{\;} n_r)}$} respresented in terms of associated roots $\gamma \in GF(2^m)$,  $\alpha_1 \in GF(2^{p_{1}})$, $\alpha_2 \in GF(2^{p_{2}})$, .... and $\alpha_r \in GF(2^{p_{r}})$ of    
minimal polynomials of \textbf{s}, \textbf{a}$_{1}$, \textbf{a}$_{2}$, .... and \textbf{a}$_{r}$  respectively.
\end{theorem} 
\begin{proof}

To prove the theorem for a generalized case of $r$ LFSRs multiplied togather, let we consider first a simple case of product of two LFSRs only.

Let $s_t\in GF(2^m)$ be a product sequence with period $n  \mid 2^m - 1$ having two constituent sequences $a_t \in GF(2^p)$ and $b_t \in GF(2^q)$ of LFSRs each defined over primitive polynomials with individual periods $n_1 = (2^p - 1)$ and $n_2 = (2^q - 1)$, where $n_1$ and $n_2$ are coprime to each other. Let \textbf{$A$} be a DFT spectra of $a_t$,  \textbf{$B$} be a DFT spectra of $b_t$ and \textbf{$S$} be a DFT spectra of $s_t$. 

Let $d$, $d_1$ and $d_2$ are degrees of non-zero spectral components i.e. \textbf{$S_k$}, \textbf{$A_{(k \mbox{\;} mod \mbox{\;} n_1)}$} and \textbf{$B_{(k\mbox{\;} mod \mbox{\;}n_2)}$} respresented in terms of associated roots $\gamma \in GF(2^m)$, $\alpha \in GF(2^p)$ and $\beta \in GF(2^q)$ of    
minimal polynomials of $s_t$, $a_t$ and $b_t$ respectively.
All roots of  minimum polynomials of \textbf{$a$}, \textbf{$b$} and \textbf{$s$} lie within their respective fields i.e. $\alpha \in GF(2^p)$, $\beta \in GF(2^q)$ and $\gamma \in GF(2^m)$ respectively.
$n_1$ and $n_2$ being coprime, $n = lcm (n_1, n_2)$.
By corollary~\ref{cor-index}, spectral components of \textbf{$S$} are non zero at all indices where corresponding spectral components of \textbf{$A$} and \textbf{$B$} are non zero. As all DFT spectral components of  \textbf{$S$} lie within $GF(2^m)$ and correspond to $\gamma^h$, where $0 \leq h \leq m-1$. Let we consider any $k$-th component of spectra of \textbf{$S$} corresponding to non zero DFT components of \textbf{$A$} and \textbf{$B$}, where we only need to prove that both non zero spectral components of \textbf{$A$} and \textbf{$B$} has one to one mapping to \textbf{$S$} through CRT. 
\\Transforming the relationship of \textbf{$s_t$}= \textbf{$a_t$}.\textbf{$b_t$}  into roots of associated polynomials of each sequence in their respective binary fields by using definitions of  $GF(2^m)$ by $\gamma^h$ ($0 \leq h \leq n$), $GF(2^p)$ by $\alpha^i$ ($0\leq i \leq n_1$) and $GF(2^q)$ by $\beta^j$ ($0 \leq j \leq n_2$), we have
\begin{equation}\label{gamma}
  \gamma^{d} = \alpha^{d}.\beta^{d} , \;\;\;  d = 0,1,2,.....,n-1
\end{equation}
As we can write \textbf{$s_t$}= \textbf{$a_{(t\; mod\; n_1)}$}.\textbf{$b_{(t \;mod\; n_2)}$}, equation~(\ref{gamma}) can be expressed as

\begin{equation}\label{gamma-mod}
  \gamma^{d} = \alpha^{d\;mod\; n_1}.\;\beta^{d\; mod\; n_2} , \;\;\;  t = 0,1,2,.....,n-1
\end{equation}
From equation~(\ref{gamma-mod}), there exists a unique mapping for $\gamma ^ d$, $\alpha^{d_1}$, and $\beta^{d_2}$ which can be computed using CRT as
\begin{eqnarray*}
      d  &\equiv& d_1\mbox{\; (mod} \mbox{\;}n_1) \\
      d  &\equiv& d_2\mbox{\; (mod} \mbox{\;}n_2)   
               \end{eqnarray*}

Mapping these facts on a product sequence having $r$ constituent sequences, it becomes trivial to see
  \begin{eqnarray*}
      d  &\equiv & d_1\mbox{\; (mod} \mbox{\;}n_1) \\
      d  &\equiv & d_2\mbox{\; (mod} \mbox{\;}n_2) \\
      \mbox{..} & \mbox{..} & \mbox{................}\\
      \mbox{..} & \mbox{..} & \mbox{................}\\
       d  &\equiv & d_r\mbox{\; (mod} \mbox{\;}n_r) 
               \end{eqnarray*}
               \end{proof}

As $GF(2^m)$ considered here is implicitly constituted by product of elements of $ GF(2^p)$ and $GF(2^q)$, convolution of $\alpha^i \in GF(2^p)$, $\beta^j \in GF(2^q)$ should result into spectral component $\gamma^h \in GF(2^m)$ ideally at each index $k$ . For convolutions in finite fields, readers may refer to~\cite{reed1975use}. However, when elements belong to different binary fields,  not much is known to us . Nevertheless, CRT based method of computing DFT components in higher binary fields from constituent DFT components in lower order fields is considered novel in this regard. Let us illustrate our results through an example.

\begin{example}\label{exm-3}
Consider a product sequence \texttt{s} having three LFSRs with primitive polynomials as $g_{1}(x) = x^{2}+x+1,\; g_{2}(x) = x^{3}+x+1  $ and $g_{3}(x) = x^{5}+x^{2}+1$. The outputs of LFSRs in this case are m-sequences, denoted as \textbf{a}$^{1}$, \textbf{a}$^{2}$ and \textbf{a}$^{3}$ respectively. Product stream \textbf{s} is obtained as 
  	\begin{equation}
  	 s_t = a^{1}_{t}.a^{2}_{t} .a^{3}_{t} \mbox{\;\;\;where\;\;} 0 \leq t \leq n-1
  	 \end{equation}
  	where period $n$ of \textbf{s}$_{t}$ in this case becomes 651 as $\mbox{lcm}(3,7,31)= 651$.
 DFT components of \textbf{a}, \textbf{b} ,  and \textbf{c} with primitive elements $ \alpha\in GF(2^2)$, $ \beta \in GF(2^3)$ and $ \delta\in GF(2^5)$ respectively are

\begin{itemize}

\item \textbf{A}$^{1}$ $ = \{0,1,1\}$
\item \textbf{A}$^{2}$  $= \{0,0,0,\beta^4, 0,\beta^2,\beta \}$
\item \textbf{A}$^{3}$ $= \{0,0,0,0,0,0,0,0,0,0,0,0,0,0,0,\delta^{29},0,0,0,0,0,0,0,\delta^{30},0,0,0,\delta^{15},0,$\\$\mbox{\;\;\;\;\;\;\;\;\;\;}\delta^{23},\delta^{27}\}$
\end{itemize}
To compute DFT of \textbf{s}, we need to compute its associated minimum polynomial through berlekamp massey algorithm which in this case is  $m(x)$=
$x^{30} + x^{25} + x^{24} + x^{20} + x^{19} + x^{17} + x^{16} + x^{13} + x^{10} + x^9 + x^8 + x^7 + x^4 + x^2 + 1$ with generator $\gamma$ $\in GF(2^{30})$.
\\ Having a complete period (651 bits) of \textbf{s}, we compute DFT through equation~\ref{DFT eq}. Corresponding to degree of minimum polynomial,  we get thirty non-zero DFT components at indices shown in Table~\ref{tab:s-gamma} below.

\begin{table}[H]
\small
\begin{center}
\caption[Sample Table]{Non Zero Spectral Points of \textbf{S}}
\begin{tabular}{|c| c| c| c| c| c| c| c| c| c| c|  }
\hline
 Index & 61 & 89 & 122 & 139 & 178 & 185 & 209 & 215 & 244 & 271   \\ \hline
  
Spectral Component & $\gamma^{492}$ & $\gamma^{387}$ & $\gamma^{333}$ &$ \gamma^{246}$ & $\gamma^{123}$ & $\gamma^{585}$ & $\gamma^{309}$ & $\gamma^{240}$ & $\gamma^{15}$ & $\gamma^{30}$   \\ \hline \hline 

Index & 278 & 325 & 356 & 370 & 395 & 418 & 430 & 433 & 461 & 488   \\ \hline
  
Spectral Component & $\gamma^{492}$ & $ \gamma^{60}$ & $\gamma^{246}$ &  $\gamma^{519}$ & $\gamma^{123}$ & $\gamma^{618}$ & $\gamma^{480}$ & $\gamma^{120}$ & $\gamma^{15}$ & $ \gamma^{30} $  \\ \hline \hline  

Index & 523 & 542 & 556 & 587 & 619 & 635 & 643 & 647 & 649 & 650  \\ \hline
  
Spectral Component & $\gamma^{387}$ & $\gamma^{60}$ & $\gamma^{333}$ & $\gamma^{519}$ &  $\gamma^{585}$ & $\gamma^{618}$ & $\gamma^{309}$  & $\gamma^{480}$ & $\gamma^{240}$ & $\gamma^{120}$  \\ \hline

\end{tabular}

\label{tab:s-gamma}
\end{center}

\end{table}

From corollary~\ref{cor-index}, non zero indices of \textbf{S} can be determined directly from knowing the individual DFTs of three LFSRs separately. For instance,
 \begin{eqnarray*}
      x  &\equiv& 1\mbox{\; (mod} \mbox{\;}3) \\
      x  &\equiv& 3\mbox{\; (mod} \mbox{\;}7)\\
      x  &\equiv& 15\mbox{\; (mod} \mbox{\;}31)
               \end{eqnarray*}
               gives result of 325 which exists amongst thirty non-zero DFT computations as well. Similarly with known spectral points of \textbf{A}$^{1}_{1}$ = $\alpha ^0$, \textbf{A}$^{2}_{3}$ = $\beta ^4$ and \textbf{A}$^{3}_{15}$ = $\delta^{29}$, spectral component \textbf{S}$_{325}$ can be determined directly by theorem~\ref{CRT-Mtheorem} as

\begin{eqnarray*}
      d  &\equiv& 0\mbox{\; (mod} \mbox{\;}3) \\
      d  &\equiv& 4\mbox{\; (mod} \mbox{\;}7)\\
      d  &\equiv& 29\mbox{\; (mod} \mbox{\;}31)
               \end{eqnarray*}

CRT gives the result of 60. So the spectral componenet \textbf{S}$_{325}$ $\in GF(2^{30})$ becomes $\gamma^{60}$. Similarly all non zero points of \textbf{S} $\in GF(2^{30})$ can be computed directly by theorem~\ref{CRT-Mtheorem} without the requirement of minimum polynomial $m(x)$, $n$ number of bits of \textbf{s} and classical computations of DFT by equation~\ref{DFT eq}.
 Conversely, from known DFT spectra of \textbf{S} only, individual DFT spectral points of  \textbf{A}$^1$, \textbf{A}$^2$ and \textbf{A}$^3$ can also be computed. For instance, having known $\gamma^{492}$ at \mbox{$S_{61}$}, \textbf{A}$^{1}_{1}$ is directly computed as $\alpha^0$, \textbf{A}$^{2}_{5}$ is computed as $\beta^2$ and \textbf{A}$^{3}_{30}$ is  computed as $\delta^{27}$. These results are considered very useful in cryptanalysis of LFSR based sequences.
\end{example}
 
\subsection{Generic Combinatorial Sequences}
Having considered the product sequences of multiple LFSRs, generic case of combinatorial sequences is discussed now where outputs of multiple LFSRs are combined through a non linear function involving multiplication and addition of bits in $ GF(2)$. From the established fact of theorem~\ref{CRT-Mtheorem} for product sequenecs, we now generalize the case for combinatorial generators here.

Consider a combinatorial generator consisting of $r$ constituent LFSRs. Let $z_t\in GF(2^m)$ be the output sequence of generator with period $n  \mid 2^m - 1$ and $m(x)$ be the associated minimum polynomial. Let $\gamma \in GF(2^m)$ be the root of $m(x)$, $\alpha_1 \in GF(2^{p_{1}})$, $\alpha_2 \in GF(2^{p_{2}})$, .... and $\alpha_r \in GF(2^{p_{r}})$ of    
minimal polynomials of \textbf{z}, \textbf{a}$_{1}$, \textbf{a}$_{2}$, .... and \textbf{a}$_{r}$  respectively.   The nonlinear function $f(x_1,x_2,...,x_{r-1})$ combines outputs of \textbf{r} LFSRs and produces the resultant stream \textbf{z} as
\begin{equation}
	z_t = f(a^{1}_{t}, a^{2}_{t},...a^{r}_{t}) \mbox{\;\;\;where\;\;} 0 \leq t \leq n-1
\end{equation}
As $f(x)$ is not only a product function, we have 
\begin{equation}
m(x) \neq m_{1}(x).m_{2}(x)....m_{r}(x) \\
\end{equation}
\begin{equation}
\Rightarrow \gamma \neq \alpha^{1}.\alpha^{2}....\alpha^{r}
\end{equation}

To take DFT of \textbf{z} by Equation~\ref{DFT eq}, we require $n$ bits of \textbf{z} and DFT will be computed with respect to $\gamma \in GF(2^m)$ having order $n$. Non zero DFT terms termed as linear span of \textbf{z} will be equilavalent to degree of associated minimum polynomial $m(x)$. These results are consistent to known theory of DFT in binary fields. However, few additional results are noted which are correlated to CRT based fixed patterns of sequences. 

If we take DFT of \textbf{z} with respect to generator $\sigma$ of its minimum polynomial $m(x)$,  experimental results reveal that irrespective of combining function $f(x)$, a fixed relationship between frequency components of \textbf{Z} and individual spectral components of \textbf{A}$^{1}$,\textbf{A}$^{2}$,....,\textbf{A}$^{r}$ exists at all those indices of \textbf{Z} where corresponding spectral components of \textbf{A}$^{1}$,\textbf{A}$^{2}$,....,\textbf{A}$^{r}$ are all non zero. Let we represent $d$, $d_1$, $d_2$,..., $d_r$  as degrees of non-zero spectral components of \textbf{$Z_k$}, \textbf{$A^{1}_{(k \mbox{\;} mod \mbox{\;} n_1)}$},..., \textbf{$A^{r}_{(k \mbox{\;} mod \mbox{\;} n_r)}$} respresented in terms of associated roots of respective minimum polynomials.  At any index $k$, where all corresponding spectral components of LFSR sequenecs are non zero, \textbf{Z}$_k \in GF(2^m)$ can be directly determined using through CRT as described in theorem~\ref{CRT-Mtheorem}. Similarly, from corollary~\ref{cor-index} and ~\ref{cor-CRT}, non zero indices of \textbf{S} corresponding to non zero spectral components of \textbf{A}$^{1}$,\textbf{A}$^{2}$,....,\textbf{A}$^{r}$ are directly determined. Let we validate our observations through an example of a simple combiner generator.

\begin{example}\label{ex-combiner}
With the same assumptions of Example~\ref{exm-3} with three LFSRs and notations used therein, output stream \textbf{z}$_t$ of a combiner is obtained as
\begin{equation}
z_t = a^{1}_{t}.a^{2}_{t} + a^{2}_{t}.a^{3}_{t} + a^{3}_{t}.a^{1}_{t} \mbox{\;\;\;where\;\;} 0 \leq t \leq n-1
\end{equation}

Taking the DFT of \textbf{z} by Equation~\ref{DFT eq} with respect to $\gamma \in GF(2^{30})$ as a generator of ($x^{30} + x^{25} + x^{24} + x^{20} + x^{19} + x^{17} + x^{16} + x^{13} + x^{10} + x^9 + x^8 + x^7 + x^4 + x^2 + 1$) with generator $\gamma$ $\in GF(2^{30})$, 31 non zero DFT points are mentioned in Table~\ref{tab:combiner-gamma} below.

\begin{table}[H]
\small
\begin{center}
\caption[Sample Table]{Non Zero Spectral Points of \textbf{Z}}
\begin{tabular}{|c| c| c| c| c| c| c| c| c| c| c|  }
\hline
 Index & 27 & 31 & 54 & 62 & 77 & 91 & 108 & 124& 153 & 156   \\ \hline
  
Spectral Component & $\gamma^{15}$ & $\gamma^{186}$ & $\gamma^{30}$ &$ \gamma^{372}$ & $\gamma^{123}$ & $\gamma^{60}$ & $\gamma^{309}$ & $\gamma^{93}$ & $\gamma^{519}$ & $\gamma^{30}$   \\ \hline \hline 

Index & 182 & 201 & 213 & 216 & 248 & 306 & 308 & 339 & 341 & 364   \\ \hline
  
Spectral Component & $\gamma^{492}$ & $ \gamma^{618}$ & $\gamma^{480}$ &  $\gamma^{120}$ & $\gamma^{186}$ & $\gamma^{618}$ & $\gamma^{387}$ & $\gamma^{333}$ & $\gamma^{93}$ & $ \gamma^{30} $  \\ \hline \hline  

Index & 371 & 402 & 426 & 432 & 495 & 496 & 511 & 573 & 581 & 612  \\ \hline
  
Spectral Component & $\gamma^{387}$ & $\gamma^{585}$ & $\gamma^{309}$ & $\gamma^{240}$ &  $\gamma^{492}$ & $\gamma^{372}$ & $\gamma^{309}$  & $\gamma^{246}$ & $\gamma^{240}$ & $\gamma^{123}$  \\ \hline

Index & 616 &  &  &  &  &  &  &  &  &   \\ \hline
  
Spectral Component & $\gamma^{387}$ &  &  &  &   &  &   &  &  &   \\ \hline

\end{tabular}

\label{tab:combiner-gamma}
\end{center}

\end{table}

Linear complexity of $z_{t}$ is determined to be 31 through berlekamp-massey algorithm and the corresponding minimum polynomial $m(x)$ in this case is: \\
 	($x^{31} + x^{29} + x^{28} + x^{27} + x^{24} + x^{23} + x^{22} + x^{20} + x^{18} + x^{17} + x^{16} + x^{15} + x^{13} + x^{11} + x^{10} + x^9 + x^8 + x^7 + x^5 + x^4 + x^2 + x + 1 $).
Now DFT is taken with respect to generator $\sigma \in GF(2^{31})$ of $m(x)$ with order 651. We will only mention spectral points at those indices where constituent LFSR sequences have all non zero spectral points.
\begin{table}[H]
\small
\begin{center}
\caption[Sample Table]{Non Zero Spectral Points of \textbf{Z} with  element $\sigma$ of $m(x)$}
\begin{tabular}{|c| c| c| c| c| c| c| c| c| c| c|  }
\hline
 Index & 61 & 89 & 122 & 139 & 178 & 185 & 209 & 215 & 244 & 271   \\ \hline
  
Spectral Component & $\sigma^{492}$ & $\sigma^{387}$ & $\sigma^{333}$ &$ \sigma^{246}$ & $\sigma^{123}$ & $\sigma^{585}$ & $\sigma^{309}$ & $\sigma^{240}$ & $\sigma^{15}$ & $\sigma^{30}$   \\ \hline \hline 

Index & 278 & 325 & 356 & 370 & 395 & 418 & 430 & 433 & 461 & 488   \\ \hline
  
Spectral Component & $\sigma^{492}$ & $ \sigma^{60}$ & $\sigma^{246}$ &  $\sigma^{519}$ & $\sigma^{123}$ & $\sigma^{618}$ & $\sigma^{480}$ & $\sigma^{120}$ & $\sigma^{15}$ & $ \sigma^{30} $  \\ \hline \hline  

Index & 523 & 542 & 556 & 587 & 619 & 635 & 643 & 647 & 649 & 650  \\ \hline
  
Spectral Component & $\sigma^{387}$ & $\sigma^{60}$ & $\sigma^{333}$ & $\sigma^{519}$ &  $\sigma^{585}$ & $\sigma^{618}$ & $\sigma^{309}$  & $\sigma^{480}$ & $\sigma^{240}$ & $\sigma^{120}$  \\ \hline

\end{tabular}

\label{tab:combiner-sigma}
\end{center}

\end{table}
Now we apply our observations on CRT based fixed patterns in sequences and compute spectral components of \textbf{Z} directly by using Theorem~\ref{CRT-Mtheorem}. From individual DFTs of LFSR sequences as computed in Example~\ref{exm-2}, corresponding to non zero indices of \textbf{A}$^1$, \textbf{A}$^2$ and \textbf{A}$^3$, we first determine non zero index of \textbf{Z} through CRT using Equation~\ref{cor-index} as
\begin{eqnarray*}
x  &\equiv& 2\mbox{\; (mod} \mbox{\;}3) \\
      x  &\equiv& 6\mbox{\; (mod} \mbox{\;}7)\\
      x  &\equiv& 30\mbox{\; (mod} \mbox{\;}31)
      \end{eqnarray*}
      we get index of $650$. Now we compute spectral value of \textbf{Z}$_{650}$ respresented in terms of $\sigma \in GF (2^{31})$ through CRT using Theorem~\ref{CRT-Mtheorem} as 
\begin{eqnarray*}
      d  &\equiv& 0\mbox{\; (mod} \mbox{\;}3) \\
      d  &\equiv& 1\mbox{\; (mod} \mbox{\;}7)\\
      d  &\equiv& 27\mbox{\; (mod} \mbox{\;}31)
               \end{eqnarray*}      
we get  \textbf{Z}$_{650}$=$\sigma^{120}$.   Spectral components of other non zero indices of $\textbf{Z}$ along with all values of $\sigma^i $ with order 651 are mentioned at appendix A. These results reveal that irrespective of non linear function $f(x)$, degree of spectral components corresponding  to spectra of constitunet LFSR sequences is consistent even being in different fields. For instance \textbf{Z}$_{650}$=$\sigma^{120}$ for a generalized combiner case and 
 \textbf{S}$_{650}$=$\gamma^{120}$ for a product case (Example~\ref{exm-2})   have same degree with different values of spectral components being $\sigma \in GF(2^{31})$ and $\gamma \in GF(2^{30})$.
 
\end{example}
\subsection{Complexity of CRT Based DFT Computations}
In this subsection, discussion on computational complexity of CRT based DFT calculations in comparison to classical DFT is presented. DFT in binary fields from Equation~\ref{DFT eq-poly} dictates that the complexity for computing each \textbf{S}$_k$ is equilvalent to cost for evaluating polynomial $s(x)$ at $\alpha^{-k}$~\cite{gongcloser} where $\psi = \alpha^{-k}$. In terms of exclusive-or operations, we have:-
\begin{enumerate}
\item The complexity of computing minimum polynomial of a sequence $\in GF(2)$ through berlekamp massey algorithm is $\mathcal{O}(m \; log\; m)$.
\item The complexity of multiplying two polynomials of degree $m$ is 
\begin{equation*}
\mathcal{O}(m\; log\; m \;log\;log\;m)
\end{equation*}
\item The complexity of solving system of $r$ linear equations over $GF(2^m)$ is 
\begin{equation*}
\mathcal{O}(r^{2.37}\;m\; log\; m \;log\;log\;m)
\end{equation*}
\item The complexity in terms of Xor operations for computing each \textbf{S}$_k$ using the Equation~\ref{DFT eq-poly} is 
\begin{equation*}
 \mathcal{O}((\;m\; log\; m \;log\;log\;m)[(log\;(k) + deg (s(x)) ])
 \end{equation*}

\end{enumerate}

For CRT based computations of spectral components from constituent spectral components, we will consider a case of product  of two  LFSR sequences which can be generalized for a combiner generator. Let we have two sequences \textbf{a} $\in GF(2^p)$ and \textbf{b} $\in GF(2^q)$. It is trivial to mention that  
\begin{equation*}
\mbox{Complexity of DFT (\textbf{s})} \gg \mbox{Complexity of [DFT (\textbf{a}) +  DFT (\textbf{b})]}
\end{equation*}

For each \textbf{S}$_k$, additional computational complexity for CRT is $\mathcal{O}(len(n)^2)$. As non zero terms of \textbf{S}$_k$ are equilavalent to the linear span of the sequence, thus total cost of CRT based computataional step of spectral components is $\mathcal{O}(LS(s)\;.\;len(n)^2)$, where $LS$ is linear span of the sequence \textbf{s}. Thus CRT based computations of spectral components of \textbf{s} for combiner generators are far efficient than classical methods of DFT computations in binary fields.

\section{CRT and Cryptanalysis of Combiner Generators}
In this subsection, discussion on application of our novel results on CRT based fixed patterns in cryptanalysis of combiner sequences is made. From discussion made in Section 3 on established linkage between
period of LFSR sequence, effect of left shifts of LFSR initial states and mathe-
matical rationale through CRT, let we demonstrate application of our observations on analysis of combiner generators.

\begin{example}

With same structure of combiner generator mentioned in Example~\ref{ex-combiner}, suppose we know 10 bits of keystream $u_t$ = [1011110001].
During off-line computations, we will generate 651 bits of reference stream i.e. $s_{t}$ with initial fills of all three LFSRs as '1' which comes out to be:
	\\
	$0010110101110110110110110110101010110110110010111000110110110$\\
$110010111011011110100110010111010010100101110$....................
$1010$\\$11010010111001011110001111010111$\\

Comparing the ten known bits of keystream $u_t$ = $[1011110001]$ with reference sequence $s_{t}$ , index position of known bits is determined as $k$=632. Thus 
\begin{equation*}\label{eq:shift-algo-app}
		u_i = s_{i+632},\;\;\;\;\;\;\; \forall \; i \; \geq 0.
		\end{equation*}
		
After determining index position of ten known bits of $u_t$ in reference stream $s_t$, we will determine initial states of LFSR by simply applying modular computations of CRT as follows: 	
\begin{eqnarray*}
      k_1  &\equiv& 632 \mbox{\; (mod \;} 3 ) \\
      k_2  &\equiv& 632 \mbox{\; (mod \;} 7 ) \\ 
      k_3  &\equiv& 632 \mbox{\; (mod \;} 31) \\  
               \end{eqnarray*}	
	Therefore, $k_1 \equiv 2 \mbox{\;}(\mbox{mod\;} 3)$ , $k_2 \equiv 2\mbox{\;} (\mbox{mod\;} 7)$ and $k_3 \equiv 12\mbox{\;} (\mbox{mod\;} 31)$.		
By using Equation~\ref{eq:trace}, initial states of LFSRs is determined as given in Table~\ref{tab:initial states-3LFSRs} below.

\begin{remark}

Generating the complete period of reference sequence \textbf{s}$_t$ followed by finding few known bits of available keystream \textbf{u}$_t$ in a complete period of 
\textbf{s}$_t$ may not be computationally feasible for sequence of larger periods which infact is the case of practical stream ciphers. However, the example is given to demonstrate the existing cyclic structures of LFSR based sequences designs and their CRT based interpretation. 
\end{remark}		

Now, let frequency domain analysis of combiner sequences is made in the light of our results on CRT based relevance of spectral components. With the same notataions as of Example~\ref{ex-combiner}, if spectral component \textbf{Z}$_k$ is computed from known ciphertext stream by any method where $k$ corresponds to all non zero unknown spectral componenets of constituent LFSR sequences, these individual spectral componenets are computed using Theorem~\ref{CRT-Mtheorem} and Corollary~\ref{cor-index}. With known spectral components of constituent LFSRs, initial states of LFSRs is determined by using Equation~(\ref{DFT shift eq}) and (\ref{eq:trace}). For instance, for \textbf{Z}$_{650}$=$\sigma^{101}$, we will do modular computations to determine the spectral component of individual LFSRs as 
\begin{eqnarray*}
      101  &\equiv& 2\mbox{\; (mod} \mbox{\;}3) \\
      101  &\equiv& 3\mbox{\; (mod} \mbox{\;}7)\\
      101  &\equiv& 10\mbox{\; (mod} \mbox{\;}31)
               \end{eqnarray*}

We get $\alpha^2 \in GF(2^2)$  at  \textbf{A}$^{1}_{2}$, $\beta^3 \in GF(2^3)$ at \textbf{A}$^{2}_6$ and $\gamma^{10} \in GF(2^5)$ at \textbf{A}$^{3}_{30}$. Now shift value for each LFSR is computed using Equation~(\ref{DFT shift eq}) as 
\begin{equation}\label{DFT shift eq-2}
  \alpha^{\tau} =  ( Z_{k} . S_{k}^{-1} )^{k^{-1}}
\end{equation}
 where $\tau$ determines the exact amount of shift between $s_t$ and $z_t$ and k is index of any one component of DFT spectra.
 
 Having determined the exact shift value for each LFSR, their initial states will be computed using Equation~(\ref{eq:trace}) within each subfield $GF(2^j)$ as
	\begin{eqnarray*}	
      b_{t}^{1} &= &Tr_{1}^{n} (\alpha^{*} \alpha^{t}) \\
      b_{t}^{2} &= &Tr_{1}^{n} (\beta^{*} \beta^{t}) \\
      b_{t}^{3} &= &Tr_{1}^{n} (\gamma^{*} \gamma^{t})       
			\end{eqnarray*}
where 
		\begin{eqnarray*}
      \alpha^{*}_i  &=& \alpha^{\tau_i}
       \\
      \beta^{*}_j  &=& \beta^{\tau_j}
       \\ 
       \gamma^{*}_l  &=& \gamma^{\tau_l}
        	\end{eqnarray*}
		
	 The initial fills of LFSRs with refernce to intial state of '1' for all LFSRs with 1 left shift in $a^{1}_t$, 5 left shifts in $a^{2}_t$ and 19 left shifts in $a^{3}_t$ gives:
 	
\begin{table}[H]
\begin{center}
\caption[Sample Table]{Initial States of 3 LFSRs}
\begin{tabular}{|c| c| }
\hline
& Initial State 
   \\ \hline
 LFSR-1 & 10   \\ \hline
 LFSR-2 & 101   \\ \hline
 LFSR-3 & 01111   \\ \hline

\end{tabular}

\label{tab:initial states-3LFSRs}
\end{center}
\end{table}	

 \end{example}
 \begin{remark}
Application of our results on CRT based fixed patterns in combiner sequences are valid for any configuration of non linear combining function. However, point of concern for cryptanalysis is computaion of \textbf{S}$_k$ in a typical scenerio of ciphertext only attack where limitataion of known keystream bits is always a driving factor for practability of the attack. For computations of DFT spectral component, complete period of ciphertext is required which is practically not the case for cryptnalaysis attacks. Fast discrete fourier spectra attacks~\cite{gong2011fast} provide an efficient methodology to compute particular spectral points when number of known bits are far less than the complete period of the stream. Our CRT based methodolgy can be utilized in conjunction with both the DFT finding algorithms proposed in~\cite{gong2011fast} when number of known key stream bits are equal to linear span of the sequnece or even lesser than that. Detailed results on efficiency of this proposed methodology will be presented separately.
\end{remark}

\begin{remark}
With regards to cryptanalysis attacks on combinatorial sequence generators, correlation attacks~\cite{siegenthaler1985decrypting}  and their faster variants~\cite{meier1989fast} are conisdered to be the most efficient attacks~\cite{canteaut2011stream}. Computational cost of our proposed methodology of DFT spectral points, even by employing fast discrete fourier spectra attacks,  is more than correlation attacks. However, in a scenerio of correlation immune non linear boolean functions when coorelation attacks are not succesful, our proposed methodology of CRT based spectral computataions is still valid which will be addressed at a separate forum.

\end{remark}

\section{Conclusion}
In this paper, new results on  CRT based analysis of combinatorial sequences have been  presented. We explored inherent peculiarities of the LFSR based combiner generators through novel patterns identified with the help of a CRT based approach. These findings were then extended to the product sequences and more particularly to the combinatorial generators. An effort was made to establish the mapping of different operations from time domain to frequency domain. Novel results on fixed shift patterns of LFSRs, their relationship to cyclic structures in finite fields and CRT based interpretation of these patterns have been exploited to establish direct relevance of final keystreams of combiner generators to individual LFSR sequences. Based on these CRT based fixed structures, new methodology of direct computating the spectral components of sequences in larger finite fields from constituent spectra of smaller fields is also presented. These new results on CRT based structural analysis of LFSR based combiners are demonstrated on small scale sequence generators with brief discussion on involved computational costs and practability of these techniques in cryptanalysis attacks.

\bibliographystyle{plain}
\bibliography{ff}

\end{document}